\documentclass[12pt,onecolumn,draft]{IEEEtran}

\usepackage{amsmath,amssymb,amsfonts}
\usepackage[final]{graphicx}
\newtheorem{corollary}{Corollary}

\newtheorem{proposition}{Proposition}
\newtheorem{lemma}{Lemma}

\def \M{\mathcal{M}}

\def \L {{\mathcal L }}

\def \E{{\textrm{E}}}

\def\M{ {\mathcal{M}} }

\begin{document}
%\tableofcontents
\title{\Large Statistical Characterization of $\kappa$-$\mu$ Shadowed Fading
\thanks{* J. F. Paris is with the Dept. Ingenier\'ia de Comunicaciones, Universidad de M\'alaga,
M\'alaga, Spain. This work is partially supported by the Spanish Government under project TEC2011-25473.}
\thanks{** This work has been submitted to the IEEE for possible publication.
Copyright may be transferred without notice, after which this version may no longer be accessible.} }

\author{
\vspace{3mm}
\authorblockN{ \normalsize Jos\'e F. Paris}}
%\\
%\vspace{0mm}
%\authorblockA{\normalsize Dept. of Ingenier\'\i a de Comunicaciones, E.T.S.I. de
%Telecomunicaci\'on,
%\\[0mm]
%Universidad de M\'alaga, M\'alaga,\\[0mm]
%E-29071, Spain}\\[0mm]
%\tt{paris@ic.uma.es} }

% El Abstract no debe superar las 50 palabras para un IEEE TCOM Transactions Letters !!!!

\maketitle
\begin{abstract}
This paper investigates a natural generalization of the
$\kappa$-$\mu$ fading channel in which the line-of-sight (LOS)
component is subject to shadowing. This fading distribution has a
clear physical interpretation, good analytical properties and
unifies the one-side Gaussian, Rayleigh, Nakagami-$m$, Ricean,
$\kappa$-$\mu$ and Ricean shadowed fading distributions. The three
basic statistical characterizations, i.e. probability density
function (PDF), cumulative distribution function (CDF) and moment
generating function (MGF), of the $\kappa$-$\mu$ shadowed
distribution are obtained in closed-form. Then, it is also shown
that the sum and maximum distributions of independent but
arbitrarily distributed $\kappa$-$\mu$ shadowed variates can be
expressed in closed-form. This set of new statistical results is
finally applied to the performance analysis of several wireless
communication systems.
\end{abstract}

\vspace{0mm}

\IEEEpeerreviewmaketitle

\section{Introduction}

\PARstart{T}{he} $\kappa$-$\mu$ fading distribution provides a
general multipath model for a line-of-sight (LOS) propagation scenario
controlled by two shape parameters $\kappa$ and $\mu$. Some
classical fading distributions are included in the $\kappa$-$\mu$
distribution as particular cases, e.g. one-sided Gaussian,
Rayleigh, \mbox{Nakagami-$m$} and Ricean. In fact, the fitting of the
$\kappa$-$\mu$ distribution to experimental data is better than
that achieved by the classical distributions previously mentioned.
A detailed description of the $\kappa$-$\mu$ fading model and other related models such as $\eta$-$\mu$ and $\alpha$-$\mu$ can be
found \mbox{in \cite{Yacoub07}-\cite{Yacoub07b}.}

Shadowing can be introduced in a LOS multipath fading model in two
basic ways. The first way consists on assuming that the total
power, associated to both the dominant components and the
scattered waves, is subject to random fluctuations. The second way
relies on assuming that only the dominant components are
subject to random fluctuations. The first class of composite
multipath/shadowing models are named multiplicative shadow fading
models, while the second class of models are named LOS shadow
fading models.

The Ricean shadowed fading model is a LOS shadow
fading model that assumes the Ricean distribution for the
multipath fading and the Nakagami-$m$ distribution for the
shadowing. This model fits to the land mobile satellite (LMS) channel
experimental data \cite{Alouni03}, and recently, it has been shown
that provides an excellent experimental fitting to underwater
acoustic communications (UAC) channels \cite{Paris12}. In addition,
the Ricean shadowed fading distribution has good analytical properties and all its basic
statistical characterizations are given in closed-form, i.e. the probability density function (PDF) and
moment generating function (MGF)
are given in \cite{Simon05} and the cumulative density function (CDF) in \cite{Paris10}.

Since the $\kappa$-$\mu$ distribution includes the Ricean
distribution as a particular case, a natural generalization of the
$\kappa$-$\mu$ distribution can be obtained by a LOS shadow fading
model with the same multipath/shadowing scheme used in the Ricean shadowed
model. This paper investigates this new model, \mbox{the
$\kappa$-$\mu$} shadowed distribution, so called by analogy with
the Ricean shadowed distribution. It is shown in this paper that the
$\kappa$-$\mu$ shadowed distribution has three main properties:
\begin{itemize}
\item It is motivated by a clear underlying physical model. \item
It provides a remarkable unification of popular fading models:
one-side Gaussian, Rayleigh, Nakagami-$m$, Ricean, $\kappa$-$\mu$
and Ricean shadowed. \item It has good analytical properties; its
PDF, CDF and MGF are obtained in closed-form. The statistics of
the sum and maximum distributions can also be expressed in
closed-form.
\end{itemize}

The remainder of this paper is organized as follows. In
\mbox{Section \ref{pmodel}} the $\kappa$-$\mu$ shadowed
underlaying physical model is described. The PDF, CDF and MGF of
the $\kappa$-$\mu$ shadowed distribution are derived in
\mbox{Section \ref{statistics}}, while the sum and maximum
distribution are investigated in \mbox{Section
\ref{advancedstatistics}.} Some applications of this set of
statistical results to the performance analysis of wireless
communication systems are presented in \mbox{Section
\ref{numerical}}. Finally, some conclusions are given in
\mbox{Section \ref{conclusion}.}

\section{Physical Model}
\label{pmodel}

The fading model for the $\kappa$-$\mu$ shadowed distribution relies on a generalization of the physical
model corresponding to the $\kappa$-$\mu$ distribution \mbox{\cite{Yacoub07}.} We consider a signal
structured in clusters of waves which propagates in a nonhomogeneous
environment. Within each cluster the multipath waves are assumed to have scattered waves
with identical powers and a dominant component with certain arbitrary power.
While the intracluster scattered waves have random phases and similar delay times, the intercluster delay-time spreads are considered relatively large.
In contrast with the $\kappa$-$\mu$ model which
assumes a deterministic dominant component within each cluster, the $\kappa$-$\mu$ shadowed model assumes that the dominant components of
all the clusters can randomly fluctuate as a consequence of shadowing.

From the physical model for the $\kappa$-$\mu$ shadowed distribution, the signal power $W$
can be expressed in terms of the in-phase and quadrature components of the fading signal
as follows
\begin{equation}
\label{model}
W = \sum\limits_{i = 1}^n {(X_i  + \xi p_i )^2 }  + \sum\limits_{i = 1}^n {(Y_i  + \xi q_i )^2 ,}
\end{equation}
where $n$ is a natural number, $X_i$ and $Y_j$ are mutually independent Gaussian processes with $\E[X_i]=\E[Y_i]=0$, $\E[X_i^2]=\E[Y_i^2]=\sigma^2$,
$p_i$ and $q_j$ are real numbers,
and $\xi$ is a \mbox{Nakagami-$m$} random variable with shaping \mbox{parameter $m$} and $\E[\xi^2]=1$.
The interpretation of (\ref{model}) is the following. Each multipath cluster is modelled by one term of the sum;
thus, $n$ is the number of multipath clusters. The scattered components of the $i$th cluster are represented by the
circularly symmetric complex Gaussian random variable $X_i+j Y_i$. For each cluster the total power of the scattered components \mbox{is $2\sigma^2$.}
The dominant component of the $i$th cluster is a complex random variable given by $\xi p_i+j \xi q_i$; thus,
its power is given by $p_i^2+q_i^2$. All the dominant components are subject to the same common shadowing fluctuation, represented
by the power normalized random amplitude $\xi$.

\section{Fundamental Statistics}
\label{statistics}

This section includes the derivation of the PDF, CDF and MGF of the $\kappa$-$\mu$ shadowed distribution.
For the sake of brevity, both the distribution of the signal envelope and the distribution of the signal power (or equivalently the instantaneous signal-to-noise ratio (SNR))
will be named $\kappa$-$\mu$ shadowed; both distributions are connected by a quadratic transformation.

The model represented in (\ref{model}) implies that
the conditional probability of the signal power $W$ given the shadowing amplitude $\xi$ follows a $\kappa$-$\mu$ distribution with PDF \cite{Yacoub07}
\begin{equation}
\label{PDF1}
f_{W\left| \xi  \right.} \left( w;\xi \right) = \frac{1}
{{2\sigma ^2 }}\left( {\frac{w}
{{\xi ^2 d^2 }}} \right)^{\frac{{n  - 1}}
{2}} e^{ - \frac{{w + d^2 }}
{{2\sigma ^2 }}} I_{n  - 1} \left( {\frac{{\xi d}}
{{\sigma ^2 }}\sqrt w } \right),
\end{equation}
where $d^2=\sum_{i=1}^{n} p_i^2+q_i^2$
represents the mean power of the dominant components and $I_{\nu}(\cdot)$ is the modified Bessel function of the
first kind \cite{Gradstein2007}. As noted in
[1], the natural number of \mbox{clusters $n$} can be replaced in
(\ref{PDF1}) by the nonnegative real extension $\mu$, resulting in
a more general and flexible distribution. The $\kappa$ parameter
is then defined as $\kappa=d^2/(2\sigma^2\mu)$ and can be
interpreted, when $\mu$ is a natural number, as the ratio between
the total power of the dominant components and the total power of
the scattered waves. In many practical analyses, the random
variable $\gamma$ representing the instantaneous SNR is used to
model the fading channel; thus, hereinafter we will consider the
random variable $\gamma\triangleq\bar\gamma W/\bar W$, where
$\bar\gamma\triangleq\E[\gamma]$ and $\bar
W=\E[W]=d^2+2\sigma^2\mu$. In terms of the scaled random variable
$\gamma$, the conditional PDF in (\ref{PDF1}) can be rewritten as
\begin{equation}
\label{PDF2}
\begin{gathered}
  f_{\gamma \left| \xi  \right.} \left( {\gamma ;\xi } \right) = \frac{{\mu \left( {1 + \kappa } \right)^{\frac{{\mu  + 1}}
{2}} }}
{{\bar \gamma \kappa ^{\frac{{\mu  - 1}}
{2}} e^{\xi ^2 \mu \kappa } }}\left( {\frac{\gamma }
{{\xi ^2 \bar \gamma }}} \right)^{\frac{{\mu  - 1}}
{2}}e^{ - \frac{{\mu \left( {1 + \kappa } \right)\gamma }}
{{\bar \gamma }}} I_{\mu  - 1} \left( {2\mu \xi \sqrt {\frac{{\kappa \left( {1 + \kappa } \right)\gamma }}
{{\bar \gamma }}} } \right). \hfill \\
\end{gathered}
\end{equation}

Let $\gamma$ be a $\kappa$-$\mu$ shadowed random variable with mean $\bar\gamma$ and real nonnegative shaping parameters $\kappa$, $\mu$ and $m$.
This fact is expressed symbolically as $\gamma\sim{\mathcal S}_{\kappa\mu}(\bar\gamma;\kappa,\mu, m)$.
The PDF of the $\kappa$-$\mu$ shadowed distribution is obtained from (\ref{PDF2}) as follows.

\vspace*{3mm}
\begin{lemma}
\label{lema1}
Let $\gamma\sim{\mathcal S}_{\kappa\mu}(\bar\gamma;\kappa,\mu, m)$, then, its PDF is given by
\begin{equation}
\label{eq_lema1}
\begin{gathered}
  f_\gamma  \left( \gamma  \right) = \frac{{\mu ^\mu  m^m \left( {1 + \kappa } \right)^\mu  }}
{{\Gamma \left( \mu  \right)\bar \gamma \left( {\mu \kappa  + m} \right)^m }} \left( {\frac{\gamma }
{{\bar \gamma }}} \right)^{\mu  - 1} e^{ - \frac{{\mu \left( {1 + \kappa } \right)\gamma }}
{{\bar \gamma }}} {}_1F_1 \left( {m,\mu ;\frac{{\mu ^2 \kappa \left( {1 + \kappa } \right)}}
{{\mu \kappa  + m}}\frac{\gamma }
{{\bar \gamma }}} \right), \hfill \\
\end{gathered}
\end{equation}
where $_1F_1(\cdot)$ is the confluent hypergeometric function \cite{Gradstein2007}.
\end{lemma}

\begin{proof}
See Appendix I.
\end{proof}
\vspace*{3mm}

It can be checked that the PDF in Lemma \ref{lema1} is the Ricean shadowed fading PDF when $\mu=1$. Next, the MGF of the
$\kappa$-$\mu$ shadowed distribution is derived from its PDF.

\vspace{3mm}
\begin{lemma}
\label{lema2}
Let $\gamma\sim{\mathcal S}_{\kappa\mu}(\bar\gamma;\kappa,\mu, m)$; then, its CDF is given by
\begin{equation}
\label{eq_lema2}
\M_\gamma  \left( s \right) \triangleq \E[e^{\gamma s} ] = \frac{{\left( { - \mu } \right)^\mu  m^m \left( {1 + \kappa } \right)^\mu  }}
{{\bar \gamma ^\mu  \left( {\mu \kappa  + m} \right)^m }}\frac{{\left( {s - \frac{{\mu \left( {1 + \kappa } \right)}}
{{\bar \gamma }}} \right)^{m - \mu } }}
{{\left( {s - \frac{{\mu \left( {1 + \kappa } \right)}}
{{\bar \gamma }}\frac{m}
{{\mu \kappa  + m}}} \right)^m }}.
\end{equation}
\end{lemma}
\begin{proof}
See Appendix II.
\end{proof}
\vspace{3mm}

The CDF of the $\kappa$-$\mu$ shadowed distribution can be expressed in closed-form by the bivariate confluent hypergeometric function $\Phi_2(\cdot)$, defined in \cite{Gradstein2007}.

\vspace{3mm}
\begin{lemma}
\label{lema3}
Let $\gamma\sim{\mathcal S}_{\kappa\mu}(\bar\gamma;\kappa,\mu, m)$; then, its MGF is given by
\begin{equation}
\label{eq_lema3}
\begin{gathered}
F_\gamma  \left( \gamma  \right) = \frac{{\mu ^{\mu-1}  m^m \left( {1 + \kappa } \right)^\mu  }}
{{\Gamma \left( \mu  \right)\left( {\mu \kappa  + m} \right)^m }}\left( {\frac{1}
{{\bar \gamma }}} \right)^\mu  \gamma ^\mu\times\hfill \\
\hspace{40mm}\Phi _2 \left( {\mu  - m,m,\mu  + 1; - \frac{{\mu \left( {1 + \kappa } \right)\gamma }}
{{\bar \gamma }}, - \frac{{\mu \left( {1 + \kappa } \right)}}
{{\bar \gamma }}\frac{{m\gamma }}
{{\mu \kappa  + m}}} \right).
\end{gathered}
\end{equation}
\end{lemma}
\begin{proof}
Taking into account that $F_{\gamma}(\gamma)=\L^{-1}[\M_{\gamma}(-s)/s;\gamma]$, the result follows from (\ref{appex12}) and the Laplace transform pair given in \cite[eq. 4.24.3]{Erdelyi1954}.
\end{proof}
\vspace{3mm}

The fundamental statistics presented in Lemmas \ref{lema1}, 2 and \ref{lema3} provide a unification of a variety of important fading distributions.
Table I reflects the parameter specializations which allow us
to obtain the one-side Gaussian, Rayleigh, Nakagami-$m$, Ricean,
$\kappa$-$\mu$ and Ricean shadowed fading
distributions, from the three shaping parameters of the $\kappa$-$\mu$ shadowed distribution.
The PDF given in Lemma \ref{lema1} is plotted in fig. \ref{fig4} for different parameter combinations;
it is clearly shown the flexibility of the mathematical model represented by the expression (\ref{eq_lema1}).

\section{Sum and Maximum Distributions}
\label{advancedstatistics}

In this Section the distribution of the sum and the maximum of independent
non-identically distributed (i.n.d) $\kappa$-$\mu$ shadowed random variables are derived. These results indicate that this new
distribution has good analytical properties, and has great potential as a tool for modelling and analyzing a variety of wireless communication systems.

\subsection{Sum distribution}
\label{sum}

The sum distribution of random variables representing the SNRs in a fading channel plays a prominent role in the analysis of diversity systems and
space-time coding. In the next Proposition, the sum of independent $\kappa$-$\mu$ shadowed random variables is statistically characterized.

\vspace{3mm}
\begin{proposition}
\label{prop1} Let $\gamma_k\sim{\mathcal
S}_{\kappa\mu}(\bar\gamma_k;\kappa_k,\mu_k, m_k)$ for $k=1,...,M$,
where all the random variables are arbitrarily distributed and
mutually independent. The PDF of the sum $\gamma=\sum_{k=1}^{M}\gamma_k$
is given by

\vspace{3mm}
\begin{equation}
\label{eq_prop1}
\begin{gathered}
  f_\gamma  \left( \gamma  \right) = \left( {\frac{1}
{{\Gamma \left( {\sum\limits_{k = 1}^M {\mu _k } } \right)}}\prod\limits_{k = 1}^M {\frac{{\mu _k ^{\mu _k } m_k ^{m_k } \left( {1 + \kappa _k } \right)^{\mu _k } }}
{{\left( {\mu _k \kappa _k  + m_k } \right)^{m_k } }}\left( {\frac{1}
{{\bar \gamma _k }}} \right)^{\mu _k } } } \right)\gamma ^{\sum\limits_{k = 1}^M {\mu _k }  - 1}  \times  \hfill \\
  \Phi _2^{(2M)} \left( {\mu _1  - m_1 , \ldots ,\mu _M  - m_{M} ,m_1 , \ldots ,m_M ;\sum\limits_{k = 1}^M {\mu _k } ;} \right.\frac{{ - \mu _1 \left( {1 + \kappa _1 } \right)\gamma }}
{{\bar \gamma _1 }}, \ldots ,\frac{{ - \mu _M \left( {1 + \kappa _M } \right)\gamma }}
{{\bar \gamma _M }}, \hfill \\
  \hspace{55mm} \frac{{ - \mu _1 \left( {1 + \kappa _1 } \right)}}
{{\bar \gamma _1 }}\frac{{m_1 \gamma }}
{{\mu _1 \kappa _1  + m_1 }}, \ldots ,\frac{{ - \mu _M \left( {1 + \kappa _M } \right)}}
{{\bar \gamma _M }}\left. {\frac{{m_M \gamma }}
{{\mu _M \kappa _M  + m_M }}} \right). \hfill \\
\end{gathered}
\end{equation}
where $\Phi_2^{(N)}(\cdot)$ is the confluent multivariate hypergeometric \mbox{function \cite{Srivastava1985}.} The CDF of $\gamma$ is given by
\begin{equation}
\label{eq_prop2}
\begin{gathered}
  F_\gamma  \left( \gamma  \right) = \left( {\frac{1}
{{\Gamma \left( {1 + \sum\limits_{k = 1}^M {\mu _k } } \right)}}\prod\limits_{k = 1}^M {\frac{{\mu _k ^{\mu _k } m_k ^{m_k } \left( {1 + \kappa _k } \right)^{\mu _k } }}
{{\left( {\mu _k \kappa _k  + m_k } \right)^{m_k } }}\left( {\frac{1}
{{\bar \gamma _k }}} \right)^{\mu _k } } } \right)\gamma ^{\sum\limits_{k = 1}^M {\mu _k } }  \times  \hfill \\
  \Phi _2^{(2M)} \left( {\mu _1  - m_1 , \ldots ,\mu _M  - m_{M} ,m_1 , \ldots ,m_M ;1 + \sum\limits_{k = 1}^M {\mu _k } ;} \right.\frac{{ - \mu _1 \left( {1 + \kappa _1 } \right)\gamma }}
{{\bar \gamma _1 }}, \ldots ,\frac{{ - \mu _M \left( {1 + \kappa _M } \right)\gamma }}
{{\bar \gamma _M }}, \hfill \\
  \hspace{55mm}\frac{{ - \mu _1 \left( {1 + \kappa _1 } \right)}}
{{\bar \gamma _1 }}\frac{{m_1 \gamma }}
{{\mu _1 \kappa _1  + m_1 }}, \ldots ,\frac{{ - \mu _M \left( {1 + \kappa _M } \right)}}
{{\bar \gamma _M }}\left. {\frac{{m_M \gamma }}
{{\mu _M \kappa _M  + m_M }}} \right). \hfill \\
\end{gathered}
\end{equation}
\end{proposition}

\begin{proof}
See Appendix III.
\end{proof}
\vspace{3mm}

Once the following
technical Lemma is considered, the important independent and identically distributed (i.i.d) case for the sum distribution is obtained as a
Corollary from the previous Proposition.

\vspace{3mm}
\begin{lemma}
\label{lemma3} The confluent multivariate hypergeometric function $\Phi_2$ has the following property
\begin{equation}
\label{eq_lemma31}
\Phi _2^{(N + M)} \left( {\underbrace {\beta _1 , \ldots ,\beta _1 }_N,\underbrace {\beta _2 , \ldots ,\beta _2 }_M;\nu ;\underbrace {x_1 , \ldots ,x_1 }_N,\underbrace {x_2 , \ldots ,x_2 }_M} \right) = \Phi _2 \left( {\beta _1 N,\beta _2 M;\nu ;x_1 ,x_2 } \right),
\end{equation}
\end{lemma}
\vspace{1mm}
where $N$ and $M$ are natural numbers, $\Re[\nu]>0$, $\Re[x_1]<0$ and $\Re[x_2]<0$.

\begin{proof}
See Appendix IV.
\end{proof}
\vspace{3mm}

\begin{corollary}
\label{cor1} Let $\gamma_k\sim{\mathcal
S}_{\kappa\mu}(\bar\gamma;\kappa,\mu, m)$ for $k=1,...,M$,
i.e. all the random variables are identically distributed and
mutually independent. The PDF of the sum $\gamma=\sum_{k=1}^{M}\gamma_k$
is given by
\begin{equation}
\label{eq_cor1}
\begin{gathered}
  f_\gamma  \left( \gamma  \right) = \frac{1}
{{\Gamma \left( {M\mu } \right)}}\frac{{\mu ^{\mu M} m^{mM} \left( {1 + \kappa } \right)^{\mu M} }}
{{\left( {\mu \kappa  + m} \right)^{mM} }}\left( {\frac{1}
{{\bar \gamma }}} \right)^{\mu M} \gamma ^{M\mu  - 1}  \times  \hfill \\
  \hspace{30mm}\Phi _2 \left( {\mu M - mM,mM;\mu M;\frac{{ - \mu \left( {1 + \kappa } \right)\gamma }}
{{\bar \gamma }},\frac{{ - \mu \left( {1 + \kappa } \right)}}
{{\bar \gamma }}\frac{{m\gamma }}
{{\mu \kappa  + m}}} \right). \hfill \\
\end{gathered}
\end{equation}
The CDF of $\gamma$ is given by
\begin{equation}
\label{eq_cor2}
\begin{gathered}
  F_\gamma  \left( \gamma  \right) = \frac{1}
{{\Gamma \left( {1 + M\mu } \right)}}\frac{{\mu ^{\mu M} m^{mM} \left( {1 + \kappa } \right)^{\mu M} }}
{{\left( {\mu \kappa  + m} \right)^{mM} }}\left( {\frac{1}
{{\bar \gamma }}} \right)^{\mu M} \gamma ^{M\mu }  \times  \hfill \\
  \hspace{30mm}\Phi _2 \left( {\mu M - mM,mM;1 + \mu M;\frac{{ - \mu \left( {1 + \kappa } \right)\gamma }}
{{\bar \gamma }},\frac{{ - \mu \left( {1 + \kappa } \right)}}
{{\bar \gamma }}\frac{{m\gamma }}
{{\mu \kappa  + m}}} \right). \hfill \\
\end{gathered}
\end{equation}
\end{corollary}
\vspace{1mm}
\begin{proof}
This result is a direct consequence of Proposition \ref{prop1} and Lemma \ref{lemma3}.
\end{proof}
\vspace{3mm}

The asymptotic behavior of the PDF and CDF of the $\kappa$-$\mu$ shadowed distribution is summarized in the following result.

\vspace{3mm}
\begin{corollary}
\label{cor2} Let $\gamma_k\sim{\mathcal
S}_{\kappa\mu}(\bar\gamma_k;\kappa_k,\mu_k, m_k)$ for $k=1,...,M$,
where all the random variables are arbitrarily distributed and
mutually independent. The asymptotic behavior of the PDF of the sum $\gamma=\sum_{k=1}^{M}\gamma_k$ when $\bar\gamma_k\rightarrow \infty$ for all $k$ is given by
\begin{equation}
\label{eqcor21}
f_\gamma  \left( \gamma  \right) \sim \left( {\frac{1}
{{\Gamma \left( {\sum\limits_{k = 1}^M {\mu _k } } \right)}}\prod\limits_{k = 1}^M {\frac{{\mu _k ^{\mu _k } m_k ^{m_k } \left( {1 + \kappa _k } \right)^{\mu _k } }}
{{\left( {\mu _k \kappa _k  + m_k } \right)^{m_k } }}\left( {\frac{1}
{{\bar \gamma _k }}} \right)^{\mu _k } } } \right)\gamma ^{\sum\limits_{k = 1}^M {\mu _k }  - 1},
\end{equation}
and the asymptotic behavior of the CDF is given by
\begin{equation}
\label{eqcor22}
F_\gamma  \left( \gamma  \right) \sim \left( {\frac{1}
{{\Gamma \left( {1 + \sum\limits_{k = 1}^M {\mu _k } } \right)}}\prod\limits_{k = 1}^M {\frac{{\mu _k ^{\mu _k } m_k ^{m_k } \left( {1 + \kappa _k } \right)^{\mu _k } }}
{{\left( {\mu _k \kappa _k  + m_k } \right)^{m_k } }}\left( {\frac{1}
{{\bar \gamma _k }}} \right)^{\mu _k } } } \right)\gamma ^{\sum\limits_{k = 1}^M {\mu _k } }.
\end{equation}
\end{corollary}
\vspace{1mm}
\begin{proof}
This result is a direct consequence of Proposition \ref{prop1} and the following trivial fact $\Phi_2^{(N)}(\beta_1,...,\beta_N;\nu;0,...,0)=1$.
\end{proof}
\vspace{3mm}

\subsection{Maximum distribution}
\label{max}

The statistical characterization of the maximum of $M$ independent $\kappa$-$\mu$ shadowed random variables is straightforward from the previous results. In general,
the CDF and the PDF for such maximum are respectively given by
\begin{equation}
\label{eq_max1}
\left\{ \begin{gathered}
  F_{\max \left\{ {\gamma _k } \right\}} \left( \gamma  \right) = \prod\limits_{k = 1}^M {F_{\gamma _k } \left( \gamma  \right)} , \hfill \\
  f_{\max \left\{ {\gamma _k } \right\}} \left( \gamma  \right) = \frac{d}
{{d\gamma }}\prod\limits_{k = 1}^M {F_{\gamma _k } \left( \gamma  \right)}  = \sum\limits_{k = 1}^M {f_{\gamma _k } \left( \gamma  \right)\prod\limits_{r = 1,r \ne k}^M {F_{\gamma _k } \left( \gamma  \right)} },  \hfill \\
\end{gathered}  \right.
\end{equation}
where $f_{\gamma _k }$ and $F_{\gamma _k }$ for $k=1,...,M$ are the corresponding marginal PDFs and CDFs. Substitution of the expressions for such marginal distributions
derived in Section \ref{statistics} in (\ref{eq_max1}) provides closed-form expressions for the PDF and CDF of the maximum of independent $\kappa$-$\mu$ shadowed random variables.

\section{Performance Analysis of Wireless Communication Systems}
\label{numerical}

This Section shows that the $\kappa$-$\mu$ shadowed distribution is an useful tool for modelling and analyzing wireless communication systems.

In previous Sections it was proved
that the $\kappa$-$\mu$ shadowed fading model is a natural generalization of the $\kappa$-$\mu$ model and unifies a variety of popular fading models. Since the $\kappa$-$\mu$
shadowed fading model has an additional parameter $m$ with respect to the $\kappa$-$\mu$ model which is physically related to shadowing;
the fitting of experimental data to the $\kappa$-$\mu$ shadowed model must be as least as good as the fitting to the $\kappa$-$\mu$ model. Otherwise, the same statement is applicable to
the Ricean shadowed model due to the $\kappa$-$\mu$ shadowed model has an extra shaping parameter $\mu$ with respect to the Ricean shadowed model.
Both the $\kappa$-$\mu$ model and the Ricean shadowed model have been proved very useful to model fading scenarios as diverse as mobile radio communications,
land mobile satellite communications and underwater acoustic communications \cite{Yacoub07}-\cite{Paris12}; thus, the $\kappa$-$\mu$ shadowed model which
encompasses these two models represents a very general tool to characterize fading channels.

With regard to the utility of the $\kappa$-$\mu$ shadowed model for the analysis of wireless communication systems, we will show below that the closed-form statistics derived in previous Sections
allows us to obtain closed-form expressions for certain fundamental performance metrics. In particular, the outage probability and/or the error
probability for $\kappa$-$\mu$ shadowed fading channels will be obtained when the receiver performs maximal ratio combining (MRC) or selection combining (SC).
These new expressions generalize all the results found in the literature for the $\kappa$-$\mu$ fading distribution and the Ricean shadowed distribution, and all the
fading distributions encompassed by these two models.

\subsection{Selection combining with $\kappa$-$\mu$ shadowed fading}
\label{numericalSC}

Let us consider a receiver with $L$ branches which performs SC. Each branch experiences $\kappa$-$\mu$ shadowed fading with an instantaneous SNR $\gamma_k\sim{\mathcal
S}_{\kappa\mu}(\bar\gamma_k;\kappa_k,\mu_k, m_k)$ for $k=1,...,L$. It is assumed that all the random variables $\gamma_k$ are mutually independent. Then, using (\ref{eq_max1}) and
\mbox{Lemma \ref{lema3},} the outage probability for SC is given by
\begin{equation}
\label{eq_SC1}
\begin{gathered}
  P_o  = \Pr \left\{ {\gamma _{SC}  \leqslant \eta } \right\} = \prod\limits_{k = 1}^L {\frac{{\mu _k ^{\mu _k } m_k ^{m_k } \left( {1 + \kappa _k } \right)^{\mu _k } }}
{{\Gamma \left( {\mu _k } \right)\left( {\mu _k \kappa _k  + m_k } \right)^{m_k } }}\left( {\frac{1}
{{\bar \gamma _k }}} \right)^{\mu _k } \eta ^{\mu _k }  \times }  \hfill \\
  \hspace{30mm}\Phi _2 \left( {\mu _k  - m_k ,m_k ,\mu _k  + 1; - \frac{{\mu _k \left( {1 + \kappa _k } \right)\eta }}
{{\bar \gamma _k }}, - \frac{{\mu _k \left( {1 + \kappa _k } \right)}}
{{\bar \gamma _k }}\frac{{m_k \eta }}
{{\mu _k \kappa _k  + m_k }}} \right), \hfill \\
\end{gathered}
\end{equation}
where $\eta$ is the SNR threshold. Since the function $\Phi_2$ tends to unity when $\gamma_k\rightarrow 0$ for all $k=1,...,L$.
After taking into account that $\Phi_2^{(N)}(\beta_1,...,\beta_N;\nu;0,...,0)=1$, the asymptotic behavior of $P_o$ is given by

\begin{equation}
\label{eq_SC2}
P_o  \sim \prod\limits_{k = 1}^L {\frac{{\mu _k ^{\mu _k } m_k ^{m_k } \left( {1 + \kappa _k } \right)^{\mu _k } }}
{{\Gamma \left( {\mu _k } \right)\left( {\mu _k \kappa _k  + m_k } \right)^{m_k } }}\left( {\frac{1}
{{\bar \gamma _k }}} \right)^{\mu _k } \eta ^{\mu _k } }.
\end{equation}

Fig. 2 shows the outage probability for SC computed by
(\ref{eq_SC1}), and superimposed simulation results which validate
the analytical derivations. Some comments on the numerical
computation of the $\Phi_2$ function are presented in Appendix V.
In Fig. 2 it is assumed a particular scenario with three branches
for SC in which
$\bar\gamma_1=\bar\gamma_2=\bar\gamma_3=\bar\gamma$,
$\kappa_1=1.2$, $\kappa_2=2.7$, $\kappa_3=3.1$, $\mu_1=4$,
$\mu_2=2$, $\mu_3=1$ and $m_1=m_2=m_3=m$. The curves represent the
outage probability in terms of the average SNR per branch
$\bar\gamma$ for different values of the shaping parameter $m$.
The results for this particular scenario show the significant
impact of shadowing in the system performance, despite the
$\kappa$ parameter which measures the LOS strength is below $5$ dB
at every branch. When $m\rightarrow\infty$ these results are
showing the performance of SC when fading is of $\kappa$-$\mu$
type.

\subsection{Maximal ratio combining with $\kappa$-$\mu$ shadowed fading}
\label{numericalMRC}

In this subsection we consider a receiver with $L$ branches which performs MRC. Each branch experiences $\kappa$-$\mu$ shadowed fading with an instantaneous SNR $\gamma_k\sim{\mathcal
S}_{\kappa\mu}(\bar\gamma_k;\kappa_k,\mu_k, m_k)$ for \mbox{$k=1,...,L$.} It is assumed that all the random variables $\gamma_k$ are mutually independent.
The outage probability is straightforward from Proposition \ref{prop1}
\begin{equation}
\label{eq_MRC1}
\begin{gathered}
  P_o  = \Pr \left\{ {\gamma _{MRC}  \leqslant \eta } \right\} = \left( {\frac{1}
{{\Gamma \left( {1 + \sum\limits_{k = 1}^L {\mu _k } } \right)}}\prod\limits_{k = 1}^L {\frac{{\mu _k ^{\mu _k } m_k ^{m_k } \left( {1 + \kappa _k } \right)^{\mu _k } }}
{{\left( {\mu _k \kappa _k  + m_k } \right)^{m_k } }}\left( {\frac{1}
{{\bar \gamma _k }}} \right)^{\mu _k } } } \right)\eta ^{\sum\limits_{k = 1}^L {\mu _k } }  \times  \hfill \\
  \Phi _2^{(2L)} \left( {\mu _1  - m_1 , \ldots ,\mu _L  - m_{L} ,m_1 , \ldots ,m_L ;1 + \sum\limits_{k = 1}^L {\mu _k } ;} \right.\frac{{ - \mu _1 \left( {1 + \kappa _1 } \right)\eta }}
{{\bar \gamma _1 }}, \ldots ,\frac{{ - \mu _L \left( {1 + \kappa _L } \right)\eta }}
{{\bar \gamma _L }}, \hfill \\
  \hspace{60mm}\frac{{ - \mu _1 \left( {1 + \kappa _1 } \right)}}
{{\bar \gamma _1 }}\frac{{m_1 \eta }}
{{\mu _1 \kappa _1  + m_1 }}, \ldots ,\frac{{ - \mu _L \left( {1 + \kappa _L } \right)}}
{{\bar \gamma _L }}\left. {\frac{{m_L \eta }}
{{\mu _L \kappa _L  + m_L }}} \right), \hfill \\
\end{gathered}
\end{equation}
where $\eta$ is the SNR threshold. The asymptotic behavior of the outage probability when $\gamma_k\rightarrow 0$ for all $k=1,...,L$
is directly obtained from Corollary \ref{cor2}
\begin{equation}
\label{eq_MRC2}
P_o  \sim \left( {\frac{1}
{{\Gamma \left( {1 + \sum\limits_{k = 1}^L {\mu _k } } \right)}}\prod\limits_{k = 1}^L {\frac{{\mu _k ^{\mu _k } m_k ^{m_k } \left( {1 + \kappa _k } \right)^{\mu _k } }}
{{\left( {\mu _k \kappa _k  + m_k } \right)^{m_k } }}\left( {\frac{1}
{{\bar \gamma _k }}} \right)^{\mu _k } } } \right)\eta ^{\sum\limits_{k = 1}^L {\mu _k } } .
\end{equation}

Now we will prove that the bit error probability of MRC systems under $\kappa$-$\mu$ fading can be computed in closed-form.
The bit error probability of many wireless communication systems with coherent detection is determined by
\begin{equation}
\label{eq_MRC3}
P_b  = \sum\limits_{r = 1}^R {\alpha _r \E\left[ {Q\left( {\sqrt {\beta _r \gamma } } \right)} \right]} ,
\end{equation}
where $\{\alpha_r,\beta_r\}_{r=1}^R$ are modulation dependent constants \cite{Paris10b}.
For MRC, the bit error probability can be obtained from (\ref{eq_MRC4}) after integrating by parts.
\begin{equation}
\label{eq_MRC4}
\begin{gathered}
  P_b  = \sum\limits_{r = 1}^R {\alpha _r \int_0^\infty  {Q\left( {\sqrt {\beta _r \gamma } } \right)f_{\gamma _{MRC} } } } \left( \gamma  \right)d\gamma  =  \hfill \\
  \hspace{50mm}\sum\limits_{r = 1}^R {\frac{{\alpha _r \sqrt {\beta _r } }}
{{\sqrt {8\pi } }}\int_0^\infty  {\frac{{e^{ - \frac{{\beta _r }}
{2}\gamma } }}
{{\sqrt \gamma  }}F_{\gamma _{MRC} } } } \left( \gamma  \right)d\gamma . \hfill \\
\end{gathered}
\end{equation}
Substituting (\ref{eq_lema3}) in (\ref{eq_MRC4}) and using \cite[pp. 290, eq. 55]{Srivastava1985}, the following closed-form expression is obtained

\begin{equation}
\label{eq_MRC5}
\begin{gathered}
  P_b  = \left( {\frac{{\Gamma \left( {\frac{1}
{2} + \sum\limits_{k = 1}^L {\mu _k } } \right)}}
{{\Gamma \left( {1 + \sum\limits_{k = 1}^L {\mu _k } } \right)}}\prod\limits_{k = 1}^L {\frac{{\mu _k ^{\mu _k } m_k ^{m_k } \left( {1 + \kappa _k } \right)^{\mu _k } }}
{{\left( {\mu _k \kappa _k  + m_k } \right)^{m_k } }}\left( {\frac{1}
{{\bar \gamma _k }}} \right)^{\mu _k } } } \right)\sum\limits_{r = 1}^R {\frac{{\alpha _r \sqrt {\beta _r } }}
{{\sqrt {8\pi } }}\left( {\frac{2}
{{\beta _r }}} \right)^{\frac{1}
{2} + \sum\limits_{k = 1}^L {\mu _k } }  \times }  \hfill \\
  F_D^{(2L)} \left( {\frac{1}
{2} + \sum\limits_{k = 1}^L {\mu _k } ,\mu _1  - m_1 , \ldots ,\mu _L  - m_{L} ,m_1 , \ldots ,m_L ;} \right.1 + \sum\limits_{k = 1}^L {\mu _k } ;\frac{{ - 2\mu _1 \left( {1 + \kappa _1 } \right)}}
{{\bar \gamma _1 \beta _r }}, \ldots  \hfill \\
  \quad  \ldots ,\frac{{ - 2\mu _L \left( {1 + \kappa _L } \right)}}
{{\bar \gamma _L \beta _r }},\frac{{ - 2\mu _1 \left( {1 + \kappa _1 } \right)}}
{{\bar \gamma _1 \beta _r }}\frac{{m_1 }}
{{\mu _1 \kappa _1  + m_1 }}, \ldots ,\frac{{ - 2\mu _L \left( {1 + \kappa _L } \right)}}
{{\bar \gamma _L \beta _r }}\left. {\frac{{m_L }}
{{\mu _L \kappa _L  + m_L }}} \right), \hfill \\
\end{gathered}
\end{equation}
where $F_D^{(N)}(\cdot)$ is the multivariate Lauricella \mbox{function \cite{Srivastava1985}.}

The outage probability for MRC computed by (\ref{eq_MRC1}) is plotted in fig. 3, including superimposed simulation results which validate the analytical derivations.
The numerical computation of the multivariate $\Phi_2^{(N)}$ function is discussed in Appendix V.
The same particular scenario used for fig. 2 is assumed here; i.e. MRC with three branches in which
$\bar\gamma_1=\bar\gamma_2=\bar\gamma_3=\bar\gamma$, $\kappa_1=1.2$, $\kappa_2=2.7$, $\kappa_3=3.1$, $\mu_1=4$, $\mu_2=2$, $\mu_3=1$ and $m_1=m_2=m_3=m$.
In this figure, the outage probability for MRC is plotted as a function of the average SNR per branch $\bar\gamma$ for different values of $m$.
As in the SC case, the shadowing parameter $m$ has a great influence on the system performance.

The bit error probability for MRC is plotted in fig. 4 when a BPSK modulation is used, i.e. $R=1$, $\alpha_1=1$ and $\beta_1=2$.
Fig. 4 displays both analytical results computed by (\ref{eq_MRC5}) and simulation results.
The numerical computation of the multivariate $F_D^{(N)}$ function is discussed in Appendix V.
Again, the same particular scenario used for fig. 2 and fig. 3 is assumed here.
In this figure, the bit error probability for BPSK with MRC is plotted as a function of the average SNR per branch $\bar\gamma$ for different values of $m$.
As with the outage probability, the shadowing parameter $m$ has a great impact on bit error probability.

\section{Conclusions}
\label{conclusion}

The statistics of the $\kappa$-$\mu$ shadowed fading model have been derived along this paper.
This fading distribution is a natural generalization of the
$\kappa$-$\mu$ fading channel which includes shadowing. Such fading distribution has a clear
physical interpretation, good analytical properties and unifies
the one-side Gaussian, Rayleigh, Nakagami-$m$, Ricean,
$\kappa$-$\mu$ and Ricean shadowed fading
distributions. The three basic statistical characterizations, i.e.
probability density function (PDF), cumulative distribution
function (CDF) and moment generating function (MGF), of the
$\kappa$-$\mu$ shadowed distribution are obtained in closed-form.
It is also shown
that the sum and maximum distributions of independent but arbitrarily distributed $\kappa$-$\mu$
shadowed variates can be expressed in closed-form.
The derived closed-form statistics are given in terms of the bivariate hypergeometric functions
$\Phi_2$ and $F_D$ or the multivariate functions $\Phi_2^{(N)}$ and $F_D^{(N)}$. Numerical methods
to compute these functions have been discussed.
Finally, this set of new statistical results is applied to
the performance analysis of several wireless communication systems.
In particular, the outage probability and the bit error probability for systems employing SC and MRC over
$\kappa$-$\mu$ shadowed fading channels have been investigated.

\vspace{10mm}

\appendices

\section{Proof of Lemma I}

From (\ref{PDF2}), the PDF of $\gamma$ can be computed as
\begin{equation}
\label{appex1}
\begin{gathered}
  f_\gamma  \left( \gamma  \right) = \int_0^\infty  {f_{\gamma \left| \xi  \right.} \left( {\gamma ;\xi } \right)} f_\xi  \left( \xi  \right)d\xi  =\hfill\\
  \hspace{50mm}\frac{{\mu \left( {1 + \kappa } \right)^{\frac{{\mu  + 1}}
{2}} }}
{{\bar \gamma \kappa ^{\frac{{\mu  - 1}}
{2}} }}\left( {\frac{\gamma }
{{\bar \gamma }}} \right)^{\frac{{\mu  - 1}}
{2}} e^{ - \frac{{\mu \left( {1 + \kappa } \right)\gamma }}
{{\bar \gamma }}} \frac{{m^m }}
{{\Gamma \left( m \right)}}\Theta \left( \gamma  \right), \hfill \\
\end{gathered}
\end{equation}
where
\begin{equation}
\label{appex2}
\begin{gathered}
  \Theta \left( \gamma  \right) \triangleq
  \int_0^\infty  {2e^{ - \xi ^2 \left( {\mu \kappa  + m} \right)} \xi ^{2m - \mu } I_{\mu  - 1} \left( {2\mu \xi \sqrt {\frac{{\kappa \left( {1 + \kappa } \right)\gamma }}
{{\bar \gamma }}} } \right)} d\xi . \hfill \\
\end{gathered}
\end{equation}
The quadratic transformation ($t=\xi^2$) in the integral which appears in $\Theta(\gamma)$ yields
\begin{equation}
\label{appex3}
\Theta \left( \gamma  \right) = \int_0^\infty  {e^{ - t\left( {\mu \kappa  + m} \right)} t^{m - \frac{\mu }
{2} - \frac{1}
{2}} I_{\mu  - 1} \left( {2\mu \sqrt {\frac{{\kappa \left( {1 + \kappa } \right)\gamma }}
{{\bar \gamma }}t} } \right)} dt.
\end{equation}
Sequential application of the identities \cite[eq. 4.16.20]{Erdelyi1954} and \mbox{\cite[eq. 9.220-2]{Gradstein2007}}
allows us to express $\Theta(\gamma)$ in terms of the confluent hypergeometric function $_1F_1$
\begin{equation}
\label{appex4}
\begin{gathered}
  \Theta \left( \gamma  \right) = \frac{{\Gamma \left( m \right)}}
{{\Gamma \left( \mu  \right)}}\frac{{\left( {\mu ^2 \kappa \left( {1 + \kappa } \right)} \right)^{\frac{{\mu  - 1}}
{2}} }}
{{\left( {\mu \kappa  + m} \right)^m }} \left( {\frac{\gamma }
{{\bar \gamma }}} \right)^{\frac{{\mu  - 1}}
{2}} {}_1F_1 \left( {m,\mu ;\frac{{\mu ^2 \kappa \left( {1 + \kappa } \right)}}
{{\mu \kappa  + m}}\frac{\gamma }
{{\bar \gamma }}} \right). \hfill \\
\end{gathered}
\end{equation}
The proof is completed after plugging (\ref{appex4}) in (\ref{appex1}) and performing some algebraic simplifications.

\vspace*{10mm}

\section{Proof of Lemma II}

Taking into account the linearity and the frequency shifting properties of the Laplace transform yields
\begin{equation}
\label{appex11}
\begin{gathered}
  \M_\gamma  \left( s \right) = \L\left[ {f_\gamma  \left( \gamma  \right); - s} \right] = \frac{{\mu ^\mu  m^m \left( {1 + \kappa } \right)^\mu  }}
{{\Gamma \left( \mu  \right)\left( {\mu \kappa  + m} \right)^m }}\left( {\frac{1}
{{\bar \gamma }}} \right)^\mu   \times  \hfill \\
  \hspace{70mm} \L\left[ {\gamma ^{\mu  - 1} {}_1F_1 \left( {m,\mu ;\frac{{\mu ^2 \kappa \left( {1 + \kappa } \right)}}
{{\mu \kappa  + m}}\frac{\gamma }
{{\bar \gamma }}} \right);\frac{{\mu \left( {1 + \kappa } \right)}}
{{\bar \gamma }} - s} \right]. \hfill \\
\end{gathered}
\end{equation}
The Laplace transform in (\ref{appex11}) is recorded in \cite[eq. 4.23.1]{Erdelyi1954}; thus, the MGF can be expressed as
\begin{equation}
\label{appex12}
\M_\gamma  \left( s \right) = \frac{{\mu ^\mu  m^m \left( {1 + \kappa } \right)^\mu  }}
{{\Gamma \left( \mu  \right)\left( {\mu \kappa  + m} \right)^m }}\left( {\frac{1}
{{\bar \gamma }}} \right)^\mu  \frac{{\Gamma \left( \mu  \right)}}
{{( - s)^\mu  }}\left( {1 - \frac{{\mu \left( {1 + \kappa } \right)}}
{{s\bar \gamma }}} \right)^{ - \left( {\mu  - m} \right)} \left( {1 - \frac{{\mu \left( {1 + \kappa } \right)}}
{{s\bar \gamma }}\frac{m}
{{\mu \kappa  + m}}} \right)^{ - m} ,
\end{equation}
which after some straightforward algebraic manipulations takes the form expressed in (\ref{eq_lema3}).

\vspace*{10mm}

\section{Proof of Proposition I}
The MGF of the sum distribution is given by
\begin{equation}
\label{appex31}
\begin{gathered}
  \M_\gamma  \left( s \right) = \prod\limits_{k = 1}^M {\frac{{\mu _k ^{\mu _k } m_k ^{m_k } \left( {1 + \kappa _k } \right)^{\mu _k } }}
{{\left( {\mu _k \kappa _k  + m_k } \right)^{m_k } }}\left( {\frac{1}
{{\bar \gamma _k }}} \right)^{\mu _k }  \times }  \hfill \\
  \quad \frac{1}
{{( - s)^{\mu _k } }}\left( {1 - \frac{{\mu _k \left( {1 + \kappa _k } \right)}}
{{s\bar \gamma _k }}} \right)^{ - \left( {\mu _k  - m_k } \right)} \left( {1 - \frac{{\mu _k \left( {1 + \kappa _k } \right)}}
{{s\bar \gamma _k }}\frac{{m_k }}
{{\mu _k \kappa _k  + m_k }}} \right)^{ - m_k } . \hfill \\
\end{gathered}
\end{equation}
From (\ref{appex31}), the PDF of the sum can be expressed as
\begin{equation}
\label{appex32}
\begin{gathered}
  f_\gamma  \left( \gamma  \right) = \L^{ - 1} \left[ {M_\gamma  \left( { - s} \right);\gamma } \right] = \left( {\frac{1}
{{\Gamma \left( {\sum\limits_{k = 1}^M {\mu _k } } \right)}}\prod\limits_{k = 1}^M {\frac{{\mu _k ^{\mu _k } m_k ^{m_k } \left( {1 + \kappa _k } \right)^{\mu _k } }}
{{\left( {\mu _k \kappa _k  + m_k } \right)^{m_k } }}\left( {\frac{1}
{{\bar \gamma _k }}} \right)^{\mu _k } } } \right) \times  \hfill \\
  \quad \L^{ - 1} \left[ {\frac{{\Gamma \left( {\sum\limits_{k = 1}^M {\mu _k } } \right)}}
{{s^{\sum\limits_{k = 1}^M {\mu _k } } }}\prod\limits_{k = 1}^M {\left( {1 + \frac{{\mu _k \left( {1 + \kappa _k } \right)}}
{{s\bar \gamma _k }}} \right)^{ - \left( {\mu _k  - m_k } \right)} } \prod\limits_{k = 1}^M {\left( {1 + \frac{{\mu _k \left( {1 + \kappa _k } \right)}}
{{s\bar \gamma _k }}\frac{{m_k }}
{{\mu _k \kappa _k  + m_k }}} \right)^{ - m_k } ;\gamma } } \right]. \hfill \\
\end{gathered}
\end{equation}
In such arrangement, the right side of (\ref{appex32}) can be identified with \cite[pp. 290, eq. 55]{Srivastava1985} yielding the expression for the PDF stated
in the proposition. To obtain the CDF we can observe again that
\begin{equation}
\label{appex33}
\begin{gathered}
  F_\gamma  \left( \gamma  \right) = \L^{ - 1} \left[ {M_\gamma  \left( { - s} \right)/s;\gamma } \right] = \left( {\frac{1}
{{\Gamma \left( {1 + \sum\limits_{k = 1}^M {\mu _k } } \right)}}\prod\limits_{k = 1}^M {\frac{{\mu _k ^{\mu _k } m_k ^{m_k } \left( {1 + \kappa _k } \right)^{\mu _k } }}
{{\left( {\mu _k \kappa _k  + m_k } \right)^{m_k } }}\left( {\frac{1}
{{\bar \gamma _k }}} \right)^{\mu _k } } } \right) \times  \hfill \\
  \L^{ - 1} \left[ {\frac{{\Gamma \left( {1 + \sum\limits_{k = 1}^M {\mu _k } } \right)}}
{{s^{1 + \sum\limits_{k = 1}^M {\mu _k } } }}\prod\limits_{k = 1}^M {\left( {1 + \frac{{\mu _k \left( {1 + \kappa _k } \right)}}
{{s\bar \gamma _k }}} \right)^{ - \left( {\mu _k  - m_k } \right)} } \prod\limits_{k = 1}^M {\left( {1 + \frac{{\mu _k \left( {1 + \kappa _k } \right)}}
{{s\bar \gamma _k }}\frac{{m_k }}
{{\mu _k \kappa _k  + m_k }}} \right)^{ - m_k } ;\gamma } } \right]. \hfill \\
\end{gathered}
\end{equation}
A new identification of (\ref{appex33}) with \cite[pp. 290, eq. 55]{Srivastava1985} completes the proof.

\vspace*{10mm}

\section{Proof of Lemma IV}
Let us consider the following ancillary function
\begin{equation}
\label{appex41}
\Lambda \left( t \right) = t^{\nu  - 1} \Phi _2^{(N + M)} \left( {\underbrace {\beta _2 , \ldots ,\beta _2 }_N,\underbrace {\beta _2 , \ldots ,\beta _2 }_M;\nu ;\underbrace {x_1 t, \ldots ,x_1 t}_N,\underbrace {x_2 t, \ldots ,x_2 t}_M} \right).
\end{equation}
Considering the Laplace transform of $\Lambda \left( t \right)$ which is obtained with the help of \mbox{\cite[pp. 290, eq. 55]{Srivastava1985},} performing trivial algebraic simplifications in the
transformed domain, and returning again to the $t$-domain with \cite[pp. 290, eq. 55]{Srivastava1985} yields the required property after setting $t=1$.

\vspace*{10mm}

\section{Numerical computation of the functions $\Phi_2$ and $F_D$}

Most of the results derived in this paper involve either the bivariate functions $\Phi_2$ and $F_D$ or the multivariate functions $\Phi_2^{(N)}$ and $F_D^{(N)}$.
Therefore, some comments on the numerical computation of these special functions can be useful for the reader. Each of them will be treated separately.

The bivariate hypergeometric function $F_D$ is the same as the Apell hypergeometric function $F_1$, which is implemented in the most popular scientific software packages,
e.g. MATLAB and MATHEMATICA. Therefore, its computation is straightforward by these software tools.

The bivariate confluent hypergeometric function $\Phi_2$ is defined in the popular mathematical handbook edited by Gradshteyn and Ryzhik; however, it is not yet implemented
in MATLAB and MATHEMATICA. As with the Marcum Q function which has a Bessel series representation,
the $\Phi_2$ function can be expressed as a $_1F_1$ series which is very appropriate for numerical computation
\cite[eq. 4.19]{Bry11}
\begin{equation}
\label{eqextra1}
\Phi_2(b,b';c;w;z)=\sum_{k=0}^{\infty} \frac{(b)_k}{k!(c)_k} w^k  {_1F_1(b';c+k;z)}.
\end{equation}

The multivariate hypergeometric function $F_D^{(N)}$ is not yet implemented in MATLAB and \mbox{MATHEMATICA;} however, it can be easily computed by its Euler-type representation
and standard numerical integration methods
\begin{equation}
\label{eqextra2}
F_D^{(N)}(a,b_1,...,b_N;c;x_1,...,x_N)=\frac{\Gamma(c)}{\Gamma(a)\Gamma(c-a)}\int_0^1 t^{a-1}(1-t)^{c-a-1}(1-x_1 t)^{-b_1}...(1-x_N t)^{-b_N}dt,
\end{equation}
where $\Re[c]>\Re[a]>0$. Note that this last condition is satisfied in the multivariate $F_D^{(N)}$ function which appears in (\ref{eq_MRC5}).

The multivariate confluent hypergeometric function $\Phi_2^{(N)}$ is not yet implemented in MATLAB and \mbox{MATHEMATICA;} however it can be efficiently computed by inverting
its one-dimensional Laplace transform \cite[pp. 290, eq. 55]{Srivastava1985}. Numerical methods for inverting Laplace transforms are exhaustively discussed in \cite{Cohen2007}.

\newpage

\begin{table}[t!]
\caption{Common Fading Distributions Derived from the $\kappa$-$\mu$ Shadowed Distribution}
\begin{center}
\small{\begin{tabular}{c | c | c | c | c | c | c | c | c }

\hline\hline
\multicolumn{6}{ c }{Fading Distribution} & \multicolumn{3}{ c }{Parameters of the $\kappa$-$\mu$ Shadowed Distribution}\\

\hline

%%%%%%%%%%%%%%%%%%%%%%%%%%%%%%%%%%%%%%%%%%%
\hline
\multicolumn{6}{c |}{
\begin{minipage}{5 cm}
\vspace{2mm}
\begin{center}
One-sided Gaussian
\end{center}
\vspace{0mm}
\end{minipage}
}
& \multicolumn{3}{| c }{
\begin{minipage}{5 cm}
\vspace{2mm}
$\underline{\mu}=0.5$,  $\underline{\kappa}\rightarrow 0$, $\underline{m}\rightarrow\infty$
\end{minipage}
}\\

%%%%%%%%%%%%%%%%%%%%%%%%%%%%%%%%%%%%%%%%%%%
\hline
\multicolumn{6}{c |}{
\begin{minipage}{5 cm}
\vspace{2mm}
\begin{center}
Rayleigh
\end{center}
\vspace{0mm}
\end{minipage}
}
& \multicolumn{3}{| c }{
\begin{minipage}{5 cm}
\vspace{2mm}
$\underline{\mu}=1$, $\underline{\kappa}\rightarrow 0$, $\underline{m}\rightarrow \infty$
\end{minipage}
}\\

%%%%%%%%%%%%%%%%%%%%%%%%%%%%%%%%%%%%%%%%%%%
\hline
\multicolumn{6}{c |}{
\begin{minipage}{5 cm}
\vspace{2mm}
\begin{center}
Nakagami-$m$,\\ with shaping parameter $m$
\end{center}
\vspace{0mm}
\end{minipage}
}
& \multicolumn{3}{| c }{
\begin{minipage}{5 cm}
\vspace{2mm}
$\underline{\mu}=m$, $\underline{\kappa}\rightarrow 0$, $\underline{m}\rightarrow \infty$
\end{minipage}
}\\

%%%%%%%%%%%%%%%%%%%%%%%%%%%%%%%%%%%%%%%%%%%
\hline
\multicolumn{6}{c |}{
\begin{minipage}{5 cm}
\vspace{2mm}
\begin{center}
Rician,\\ with shaping parameter $K$
\end{center}
\vspace{0mm}
\end{minipage}
}
& \multicolumn{3}{| c }{
\begin{minipage}{5 cm}
\vspace{2mm}
$\underline{\mu}=1$, $\underline{\kappa}=K$, $\underline{m}\rightarrow \infty$
\end{minipage}
}\\

%%%%%%%%%%%%%%%%%%%%%%%%%%%%%%%%%%%%%%%%%%%
\hline
\multicolumn{6}{c |}{
\begin{minipage}{5 cm}
\vspace{2mm}
\begin{center}
$\kappa$-$\mu$,\\ with shaping parameters $\kappa$ and $\mu$
\end{center}
\vspace{0mm}
\end{minipage}
}
& \multicolumn{3}{| c }{
\begin{minipage}{5 cm}
\vspace{2mm}
$\underline{\mu}=\mu$, $\underline{\kappa}=\kappa$, $\underline{m}\rightarrow \infty$
\end{minipage}
}\\

%%%%%%%%%%%%%%%%%%%%%%%%%%%%%%%%%%%%%%%%%%%
\hline
\multicolumn{6}{c |}{
\begin{minipage}{5 cm}
\vspace{2mm}
\begin{center}
Rician shadowed,\\ with shaping parameters $K$ and $m$
\end{center}
\vspace{0mm}
\end{minipage}
}
& \multicolumn{3}{| c }{
\begin{minipage}{5 cm}
\vspace{2mm}
$\underline{\mu}=1$, $\underline{\kappa}=K$, $\underline{m}=m$
\end{minipage}
}\\

\hline\hline

\end{tabular}}
\small{\begin{tabular}{l}
%{{\normalsize$^{\star}$} Note that $K\doteq \frac{\Omega}{2b_0}$, using the same notation as in \cite{Abdi03},\cite{Simon05}.}\\
\end{tabular}}
\end{center}
\end{table}

\newpage

\begin{figure}[h]
\centering\includegraphics[width=12cm]{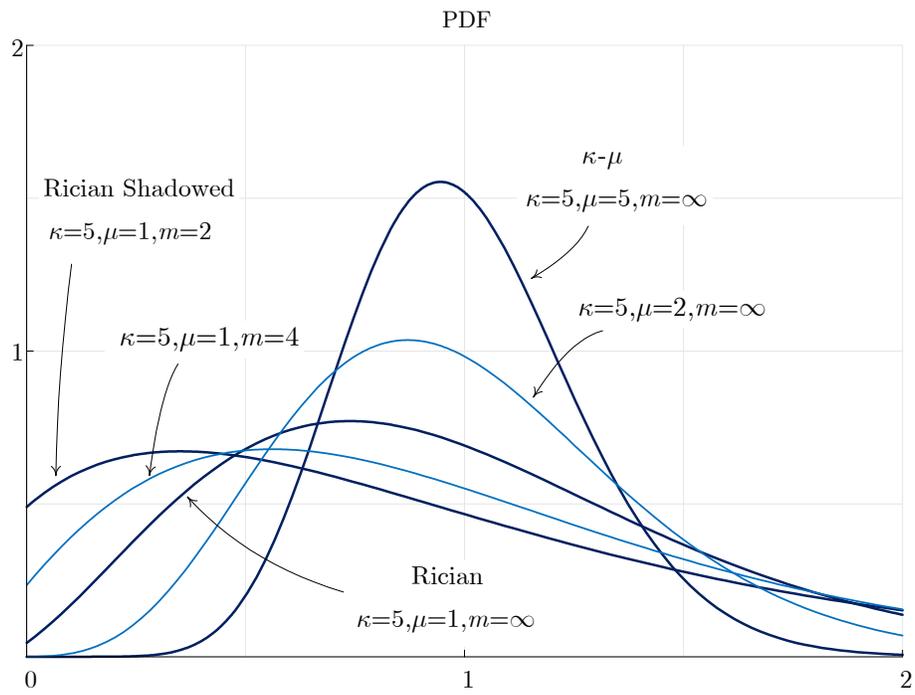}
\begin{center}
\caption{ \footnotesize
PDF of the $\kappa$-$\mu$ shadowed distribution ($\bar\gamma=1$).
}
\label{fig4}
\end{center}
\end{figure}

\newpage

\begin{figure}[h]
\centering\includegraphics[width=12cm]{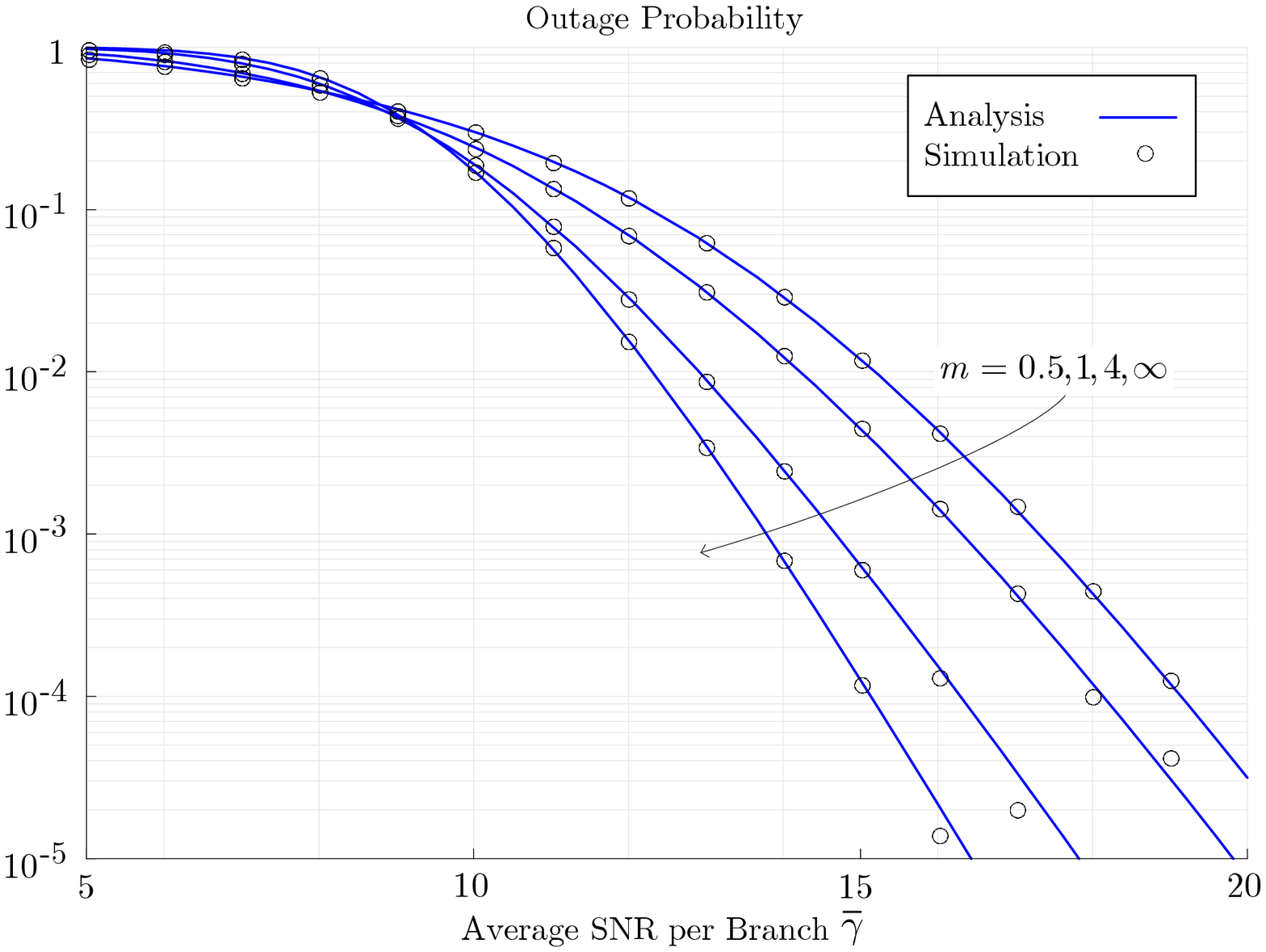}
\begin{center}
\caption{ \footnotesize
Outage probability versus average SNR per branch over $\kappa$-$\mu$ shadowed fading channels. A triple-branch SC scenario is considered, with parameters
$\bar\gamma_1=\bar\gamma_2=\bar\gamma_3=\bar\gamma$, $\kappa_1=1.2$, $\kappa_2=2.7$, $\kappa_3=3.1$, $\mu_1=4$, $\mu_2=2$, $\mu=1$ and $m_1=m_2=m_3=m$.
}
\label{fig1}
\end{center}
\end{figure}

\newpage

\begin{figure}[h]
\centering\includegraphics[width=12cm]{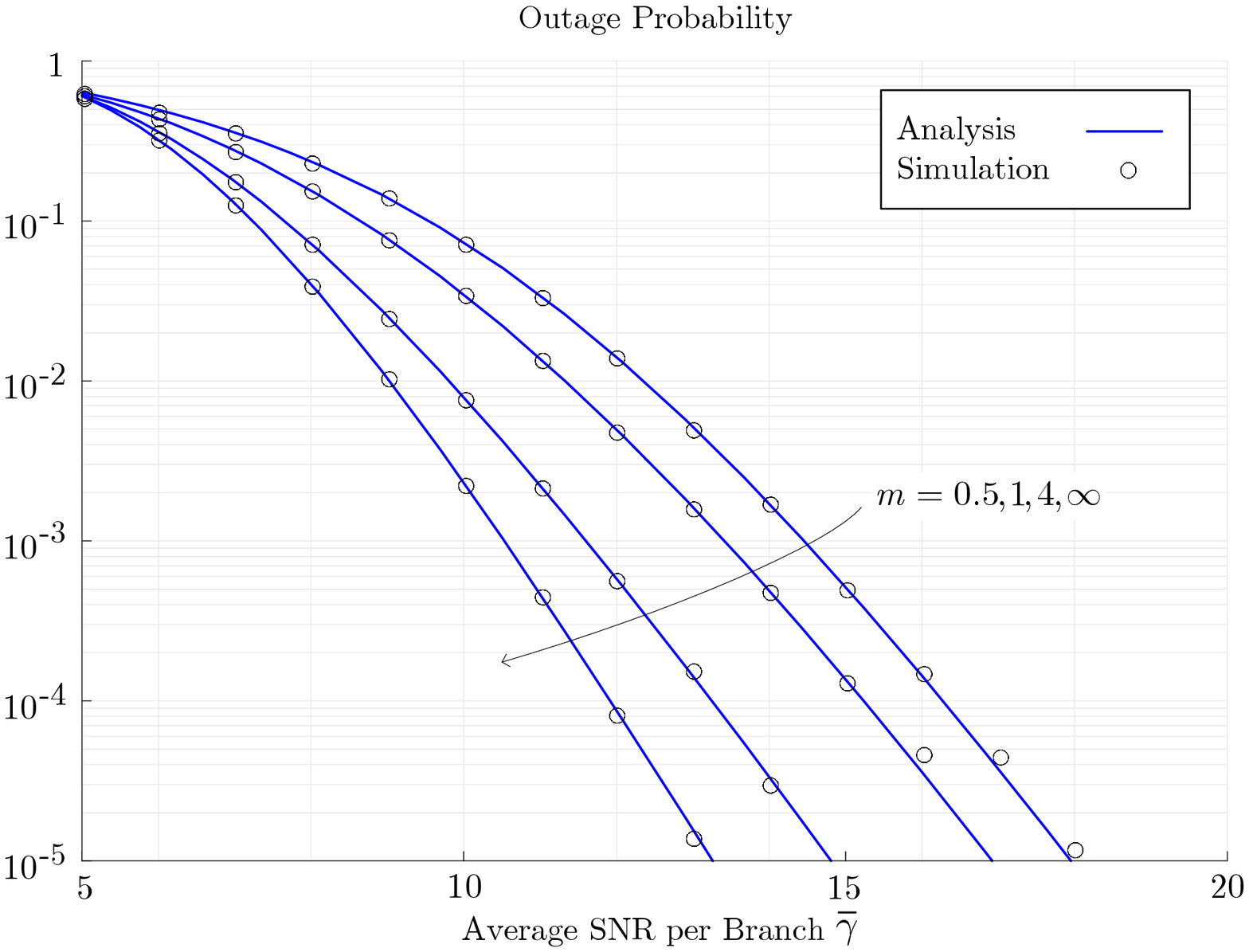}
\begin{center}
\caption{ \footnotesize
Outage probability versus average SNR per branch in $\kappa$-$\mu$ shadowed fading channels. A triple-branch MRC scenario is considered, with parameters
$\bar\gamma_1=\bar\gamma_2=\bar\gamma_3=\bar\gamma$, $\kappa_1=1.2$, $\kappa_2=2.7$, $\kappa_3=3.1$, $\mu_1=4$, $\mu_2=2$, $\mu=1$ and $m_1=m_2=m_3=m$.
}
\label{fig2}
\end{center}
\end{figure}

\newpage

\begin{figure}[h]
\centering\includegraphics[width=12cm]{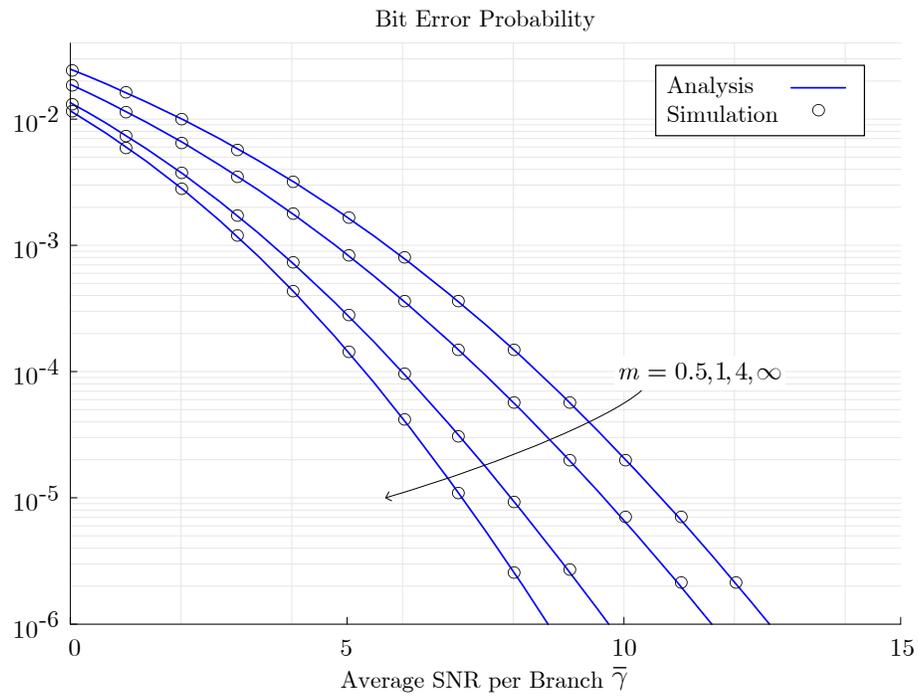}
\begin{center}
\caption{ \footnotesize
Bit error rate versus average SNR per branch in $\kappa$-$\mu$ shadowed fading channels. In this plot BPSK modulation and triple-branch MRC are considered, with parameters
$\bar\gamma_1=\bar\gamma_2=\bar\gamma_3=\bar\gamma$, $\kappa_1=1.2$, $\kappa_2=2.7$, $\kappa_3=3.1$, $\mu_1=4$, $\mu_2=2$, $\mu=1$ and $m_1=m_2=m_3=m$.
}
\label{fig2}
\end{center}
\end{figure}

\end{document}